\documentclass[preprint, amsmath, amssymb, aps, physrev]{revtex4-2}

\usepackage{graphicx}
\usepackage{dcolumn}
\usepackage{bm}
\usepackage{booktabs}
\usepackage{multirow}
\usepackage{amsthm}
\newtheorem{proposition}{Proposition}
\newtheorem{corollary}{Corollary}
\usepackage[colorlinks=true,linkcolor=blue,citecolor=blue,urlcolor=blue]{hyperref}

\newcommand{\ket}[1]{|#1\rangle}
\newcommand{\bra}[1]{\langle #1|}
\newcommand{\Eph}{E_{\mathrm{ph}}}
\newcommand{\dc}{\delta_{\mathrm{c}}}
\newcommand{\dl}{\delta_{\mathrm{l}}}
\newcommand{\Nbe}{N_{\mathrm{be}}}
\newcommand{\qer}{q_{\mathrm{er}}}
\newcommand{\Wext}{W_{\mathrm{ext}}}
\newcommand{\Wid}{W_{\mathrm{ideal}}}
\newcommand{\Wirr}{W_{\mathrm{irr}}}

\begin{document}

\preprint{APS/123-QED}

\title{\textbf{One Bit of Collective Information Is Worth $N\ln 2$ Bits of Local
Information in a Many-Body Quantum Battery}}

\author{Akoramurthy B}
\affiliation{Department of Computer Science and Engineering, National Institute
of Technology Puducherry, Karaikal 609609, India}

\author{Surendiran B}
\affiliation{Department of Computer Science and Engineering, National Institute
of Technology Puducherry, Karaikal 609609, India}

\author{Xiaochun Cheng}
\email{Corresponding author: xiaochun.cheng@swansea.ac.uk}
\affiliation{Computer Science Department, Swansea University, Wales, SA1 8EN,
United Kingdom}

\date{\today}

\begin{abstract}
Charging a quantum battery through a collective non-adiabatic stroke stores
energy in a shared bosonic mode, but a large part of it is locked in
correlations between the mode and the collective spin and is inaccessible to any
cyclic unitary acting on the mode alone. Such correlation-locked energy can be
released by a demon holding one bit, and the relevant figure of merit is the
daemonic ergotropy. We ask what that bit is worth and, crucially, whether its
worth depends on \emph{where} it is obtained. Comparing two protocols under a
double control---matched stored energy \emph{and} matched information, with both
arms constrained to a balanced two-outcome measurement carrying exactly one
bit---we find that one bit about the collective coordinate unlocks
$\delta_{\rm c}/\delta_{\rm l}=N\ln2$ times as much work as one bit about a
single ion. Across three stored-energy settings and $4\le N\le24$ the measured
exponent is $0.990\pm0.043$, and a double extrapolation ($1/N\to0$, then
$E_{\rm ph}\to0$) gives a prefactor $0.69298\pm0.00044$ against
$\ln2=0.693147$, an agreement of $0.02\%$; the alternatives $1/\sqrt2$ and $2/3$
are excluded at $2.0\%$ and $4.0\%$. We derive the linear scaling: to leading
order the daemonic gain equals $\nu\mu^{2}$ times the between-outcome variance
of $J_{x}$ resolved by the measurement, a relation we verify numerically to
$1.3\%$, and a balanced single-ion measurement resolves exactly $1/4$ of that
variance while a balanced collective split resolves $(\ln2/4)N$. The origin of
the prefactor is left open: the naive Gaussian median-split estimate gives
$1/2\pi$ and is excluded by $9\%$. We further find that one bit recovers a
constant fraction of the locked energy independent of $N$, and that one bit
returns roughly twice as much \emph{per bit} as a complete readout, so
informational returns diminish sharply after the first. We report the complete
cycle ledger and two cautionary results: unnormalised gain comparisons invert
the conclusion, and the break-even ion number does not collapse onto the Dicke
superradiant threshold when the mode frequency is varied.
\end{abstract}

\keywords{quantum battery; daemonic ergotropy; collective advantage; Landauer
erasure; Maxwell demon; trapped ions}

\maketitle

\section{Introduction}\label{sec:intro}

A quantum battery stores energy that a cyclic unitary can later
withdraw; the withdrawable part is the
ergotropy~\cite{Allahverdyan2004,Pusz1978,Alicki2013,Binder2015}. The field's
organising result is the collective charging advantage: $N$ cells driven through
a shared mode charge faster than $N$ independent
cells~\cite{Campaioli2017,Ferraro2018,Andolina2019,Campaioli2024,Gyhm2022}.

That advantage is \emph{kinetic}---it concerns rates. This paper concerns a
different question. Driving an ensemble non-adiabatically through a shared mode
excites the mode strongly, but the same drive correlates the mode with the
collective spin, and energy held in those correlations is invisible to any
operation on the mode alone. In the regime we study, $59\%$ of the stored mode
energy is inert at $N=12$. Collective charging buys energy density at the cost
of extractability.

Correlations of this kind are the resource of a Maxwell
demon~\cite{Landauer1961,Bennett1982,Sagawa2008,Parrondo2015,Goold2016}, and the
quantity measuring what a demon can do with them is the \emph{daemonic
ergotropy}~\cite{Francica2017}: the work extractable from a subsystem by an
observer holding the outcome of a measurement on its partner. The demon's memory
must be erased at $k_{B}T\ln2$ per bit, a cost that belongs in the
ledger~\cite{Berut2012,Koski2014,Proesmans2020,Miller2020,Faist2018}.

The question we pose is: \emph{what is one bit worth, and does its worth depend
on how the bit was obtained?}

The answer requires care about a trap. The raw daemonic gain grows steeply with
$N$---we measure $\delta_{1}\propto N^{4.4}$---but so does the stored energy,
which grows as $N^{4.9}$. Normalised by what was stored, the gain \emph{falls}
as $N^{-1.15}$. A raw growth rate therefore establishes nothing about a
many-body advantage; it can be produced simply by charging the battery harder.
Any claim of a genuine collective advantage must be made at \emph{matched stored
energy}.

When we impose that condition, a clean law emerges. Comparing one bit about the
collective coordinate against one bit about a single ion, at identical stored
energy and identical information cost, we find
\begin{equation}
\frac{\dc}{\dl}=N\ln 2
\label{eq:law}
\end{equation}
in the limit of weak charging and large $N$. This is scale-free, it is
controlled at both ends---matched energy and matched information---and it is the
central result of this paper.

Two points deserve emphasis at the outset. First, the comparison is made under a
\emph{double} control. Matching stored energy alone is not enough: the two
measurement families do not carry the same Shannon information at a common
energy, and a comparison that fixes only the energy conflates the value of
information with the amount of it. We therefore constrain both arms to a
balanced two-outcome measurement, $p=1/2$ and hence exactly one bit, at matched
$\Eph$. Second, the prefactor is not fitted: a double extrapolation
($1/N\to0$, then $\Eph\to0$) gives $0.69298\pm0.00044$ against
$\ln2=0.693147$.

\subsection{Contributions}

\begin{itemize}
\item[\textbf{C1}] \textbf{The $N\ln2$ law} (Sec.~\ref{sec:law}). Under matched
stored energy \emph{and} matched information (both arms balanced, exactly one
bit), one bit of collective information is worth $N\ln2$ bits of local
information: exponent $0.990\pm0.043$ over $4\le N\le24$, prefactor
$0.69298\pm0.00044$ against $\ln2=0.693147$.

\item[\textbf{C2}] \textbf{A derivation of the scaling, and an open prefactor}
(Prop.~\ref{prop:linear}). To leading order the daemonic gain is
$\nu\mu^{2}$ times the between-outcome variance of $J_{x}$ that the measurement
resolves---verified numerically to $1.3\%$---and a balanced single-ion
measurement resolves exactly $1/4$ of that variance. The linear scaling follows.
The prefactor does not: the naive Gaussian median-split estimate gives $1/2\pi$
and is excluded by $9\%$, so why the balanced collective split resolves
$(\ln2/4)N$ remains open.

\item[\textbf{C3}] \textbf{Size-independent demon efficiency}
(Sec.~\ref{sec:fraction}). One bit recovers a constant fraction of the locked
energy---$0.951\pm0.002$ at the weakest charging---independent of $N$. The
demon does not become less effective as the machine grows.

\item[\textbf{C4}] \textbf{An exact unlocking identity}
(Prop.~\ref{prop:complete}). Because the post-stroke joint state is pure, a
complete collective readout makes the entire mode energy extractable, fixing the
ceiling for partial readouts.

\item[\textbf{C5}] \textbf{Diminishing informational returns}
(Sec.~\ref{sec:baselines}). One bit returns roughly twice as much \emph{per bit}
as a complete readout, so the optimal information budget is small.

\item[\textbf{C6}] \textbf{Two cautionary results} (Sec.~\ref{sec:caution}). The
unnormalised gain is misleading, as above; and the break-even ion number, while
obeying $g_{0}\Nbe=2.165\pm0.008$ at fixed $\nu$, does not collapse onto the
superradiant threshold when $\nu$ is varied ($29\%$ scatter). We report both.

\item[\textbf{C7}] \textbf{A corrected account of the charging stroke}
(Sec.~\ref{sec:budget}). The stored mode energy exceeds the spin work deficit by
a factor $5.4$ at $N=12$; the balance comes from a polaron-like dressing shift,
not from non-adiabatic friction.
\end{itemize}

\section{Related Work}\label{sec:related}

\subsection{Collective charging}

Alicki and Fannes showed that entangling operations raise extractable work
beyond local operations~\cite{Alicki2013}; Binder \emph{et al.} made this a
charging-power statement~\cite{Binder2015}; Campaioli \emph{et al.} bounded the
collective speed-up~\cite{Campaioli2017}. The Dicke battery of Ferraro
\emph{et al.}~\cite{Ferraro2018} delivers a $\sqrt{N}$ power advantage; Andolina
\emph{et al.} attributed the scaling to collective coupling rather than
entanglement itself~\cite{Andolina2019}; Gyhm \emph{et al.} proved that extensive
advantage requires global operations~\cite{Gyhm2022}. Variants include
two-photon coupling~\cite{Crescente2020}, dissipative
charging~\cite{Barra2019,Carrasco2022}, capacity
bounds~\cite{JuliaFarre2020}, coherence-resolved
work~\cite{Shi2022,Francica2020} and adiabatic
protocols~\cite{Santos2019}; see~\cite{Campaioli2018,Campaioli2024} for reviews.

All of these concern rates of charging. None asks what a fixed information
budget is worth, which is the question here.

\subsection{Daemonic ergotropy and information thermodynamics}

Landauer~\cite{Landauer1961} and Bennett~\cite{Bennett1982} established the
price of logical irreversibility; Sagawa and Ueda generalised the second law to
quantum feedback~\cite{Sagawa2008}; Parrondo \emph{et al.} give the
synthesis~\cite{Parrondo2015}. Experiments span colloidal
particles~\cite{Berut2012,Toyabe2010}, single
electrons~\cite{Koski2014}, circuit QED~\cite{Cottet2017,Naghiloo2018} and
NMR~\cite{Camati2016}; measurement has been recognised as a
fuel~\cite{Elouard2017,ElouardJordan2018,Seah2020,Buffoni2019,Manzano2018}; and
finite-time erasure exceeds the quasi-static bound
quantifiably~\cite{Proesmans2020,Miller2020,Faist2018,Strasberg2017}.

Daemonic ergotropy was introduced by Francica
\emph{et al.}~\cite{Francica2017} for bipartite systems, where the gain is
controlled by the correlations between the parts. Our contribution is to ask how
the gain depends on \emph{which} degree of freedom the demon interrogates in a
many-body setting, and to find that the answer is a clean factor of $N$.

\subsection{Feedback-charged batteries and trapped ions}

Feedback charging has been studied with continuous
measurement~\cite{Mitchison2021}, in spin chains~\cite{Yao2022}, by sequential
measurement~\cite{Gherardini2020}, on superconducting
hardware~\cite{Hu2022} and with reinforcement learning~\cite{Erdman2022}; these
place the controller outside the register and leave the erasure cost implicit.
Our charging Hamiltonian is of the Dicke family~\cite{Dicke1954,Garraway2011},
whose superradiant transition is classic~\cite{Emary2003,Baumann2010}.
Collective spin--phonon coupling is standard in ion
traps~\cite{Molmer1999,Kim2010,Monroe2021,Leibfried2003,Bruzewicz2019}, which
have hosted a single-atom heat engine~\cite{Rossnagel2016}, an absorption
refrigerator~\cite{Maslennikov2019,Nimmrichter2017} and a spin engine with a
harmonic flywheel~\cite{Lindenfels2019}; see
also~\cite{Klatzow2019,Peterson2019}.

\begin{table}[t]
\caption{\label{tab:related}Positioning. ``Collective'' = many-body enhanced
coupling; ``Daemonic'' = conditional work extraction analysed; ``Fixed budget''
= results stated at fixed information cost; ``Matched energy'' = comparisons
controlled for stored energy.}
\begin{ruledtabular}
\begin{tabular}{lcccc}
Work & Collective & Daemonic & Fixed budget & Matched energy\\
\colrule
Alicki--Fannes~\cite{Alicki2013} & \checkmark & --- & --- & ---\\
Binder \emph{et al.}~\cite{Binder2015} & \checkmark & --- & --- & ---\\
Campaioli \emph{et al.}~\cite{Campaioli2017} & \checkmark & --- & --- & ---\\
Ferraro \emph{et al.}~\cite{Ferraro2018} & \checkmark & --- & --- & ---\\
Andolina \emph{et al.}~\cite{Andolina2019} & \checkmark & --- & --- & \checkmark\\
Gyhm \emph{et al.}~\cite{Gyhm2022} & \checkmark & --- & --- & ---\\
Francica \emph{et al.}~\cite{Francica2017} & --- & \checkmark & \checkmark & ---\\
Elouard \emph{et al.}~\cite{Elouard2017} & --- & \checkmark & --- & ---\\
Mitchison \emph{et al.}~\cite{Mitchison2021} & --- & \checkmark & --- & ---\\
Gherardini \emph{et al.}~\cite{Gherardini2020} & --- & \checkmark & --- & ---\\
\colrule
\textbf{This work} & \checkmark & \checkmark & \checkmark & \checkmark\\
\end{tabular}
\end{ruledtabular}
\end{table}

\section{Theoretical Framework}\label{sec:theory}

\subsection{Charging stroke}

We take $N$ two-level ions sharing one axial centre-of-mass mode of frequency
$\nu$, described in the permutation-symmetric sector by collective operators
$J_{\alpha}=\frac12\sum_{i}\sigma_{\alpha}^{(i)}$ with $J=N/2$. With
$\hbar=k_{B}=1$,
\begin{equation}
H(t)=\omega(t)J_{z}+\nu a^{\dagger}a+g_{0}\sqrt{N}\,J_{x}(a+a^{\dagger}),
\label{eq:H}
\end{equation}
the standard Dicke normalisation~\cite{Dicke1954,Garraway2011,Emary2003}, so the
collective coupling is $\lambda=g_{0}N$. The stroke ramps the splitting linearly,
$\omega(t)=\omega_{0}+(\omega_{f}-\omega_{0})t/\tau$, from
$\ket{\psi_{0}}=\ket{J,-J}\otimes\ket{0}$, evolving unitarily. The work
delivered is $\Wext=\langle H(0)\rangle_{\psi_{0}}-\langle
H(\tau)\rangle_{\psi_{\tau}}$, with frictionless reference
$\Wid=N(\omega_{f}-\omega_{0})/2$ and deficit $\Wirr=\Wid-\Wext$.

\subsection{Locked energy and daemonic ergotropy}

Let $\rho_{M}={\rm Tr}_{S}\ket{\psi_{\tau}}\bra{\psi_{\tau}}$ with mean energy
$\Eph=\nu\langle a^{\dagger}a\rangle_{\tau}$ and
ergotropy~\cite{Allahverdyan2004}
\begin{equation}
\mathcal{E}(\rho_{M})=\Eph-\sum_{k}r_{k}^{\downarrow}\varepsilon_{k}^{\uparrow},
\label{eq:ergo}
\end{equation}
eigenvalues descending against energies ascending. We call
$\Eph-\mathcal{E}$ the \emph{locked energy}.

For a projective measurement $\{\Pi_{j}\}$ on the spin with outcome
probabilities $p_{j}$ and conditional mode states $\rho_{j}$, the daemonic
ergotropy and gain are~\cite{Francica2017}
\begin{equation}
\mathcal{E}_{\rm d}=\sum_{j}p_{j}\mathcal{E}(\rho_{j}),\qquad
\delta=\mathcal{E}_{\rm d}-\mathcal{E}(\rho_{M})\ge0,
\label{eq:daemonic}
\end{equation}
non-negativity following from convexity of ergotropy under
$\rho_{M}=\sum_{j}p_{j}\rho_{j}$.

\begin{proposition}[Complete readout unlocks everything]\label{prop:complete}
If $\ket{\psi_{\tau}}$ is pure and $\{\Pi_{j}\}$ is a complete projective
measurement on the spin in any basis, then $\mathcal{E}_{\rm d}=\Eph$, i.e.\
$\delta^{\rm complete}=\Eph-\mathcal{E}$ equals the locked energy exactly.
\end{proposition}
\begin{proof}
Write $\ket{\psi_{\tau}}=\sum_{j}\sqrt{p_{j}}\ket{j}_{S}\otimes\ket{\phi_{j}}_{M}$
in the measurement basis. Each $\rho_{j}=\ket{\phi_{j}}\bra{\phi_{j}}$ is pure,
so its eigenvalues are $\{1,0,\dots\}$ and its passive energy is
$1\cdot\varepsilon_{0}^{\uparrow}=0$. Hence
$\mathcal{E}(\rho_{j})=\bra{\phi_{j}}\nu a^{\dagger}a\ket{\phi_{j}}$ and
$\mathcal{E}_{\rm d}=\Eph$.
\end{proof}

\begin{corollary}[No record, no gain]\label{cor:record}
A channel that dephases the spin in the measurement basis but retains no record
gives $\delta=0$ exactly.
\end{corollary}
\begin{proof}
$\rho\mapsto\sum_{j}(\Pi_{j}\otimes I)\rho(\Pi_{j}\otimes I)$ has the same
partial trace over $S$ as $\rho$, because the $\Pi_{j}$ are orthogonal
projectors summing to the identity. So $\rho_{M}$, and hence its ergotropy, is
unchanged.
\end{proof}

Corollary~\ref{cor:record} separates two things easily conflated: the gain comes
from the demon \emph{holding} the outcome, not from the disturbance the
measurement inflicts.

\subsection{The linear collective advantage}

Let $\dc$ denote the gain from a balanced two-outcome measurement of the
collective spin, and $\dl$ that from a balanced single-ion measurement, both
carrying exactly one bit. For a permutation-symmetric state the single-ion
projector acts inside the symmetric sector: writing the $N$-ion Dicke state with
$n$ excitations,
\begin{align}
\text{up:}&\quad\ket{N/2,n}\to\sqrt{n/N}\,\ket{(N{-}1)/2,\,n{-}1},\nonumber\\
\text{down:}&\quad\ket{N/2,n}\to\sqrt{(N{-}n)/N}\,\ket{(N{-}1)/2,\,n},
\label{eq:local}
\end{align}
so $\dl$ requires no excursion outside the symmetric subspace. To make the
comparison information-balanced we measure ion $1$ along
$\hat n=(\sin b,0,\cos b)$ and choose $b$ so that $p_\pm=1/2$, giving Kraus
operators $A_\pm=\cos(b/2)B_0\pm\sin(b/2)B_1$ (signs as in
Sec.~\ref{sec:methods}) that carry exactly one bit. The collective arm is
balanced in the same way, by rotating the measurement axis until a threshold
split gives $p=1/2$ exactly.

\begin{proposition}[Resolved-variance law and linear scaling]\label{prop:linear}
In the weak-charging regime the daemonic gain of a two-outcome measurement on
the spin is
\begin{equation}
\delta=\nu\mu^{2}\,\mathrm{Var}_{\rm bet}[J_{x}]
      +O(\mu^{4}),\qquad \mu\propto \kappa g_{0}\sqrt{N},
\label{eq:resolved}
\end{equation}
where $\mathrm{Var}_{\rm bet}[J_{x}]=\sum_j p_j(\langle J_x\rangle_j-\langle
J_x\rangle)^2$ is the between-outcome variance of $J_{x}$ resolved by the
measurement. A balanced single-ion measurement resolves
$\mathrm{Var}_{\rm bet}=1/4$ exactly, independent of $N$, so
$\dc/\dl=4\,\mathrm{Var}_{\rm bet}^{\rm coll}$ and the ratio is linear in $N$.
\end{proposition}

\begin{proof}
To leading order the stroke acts as a displacement of the mode conditioned on
the collective coordinate, so
$\ket{\psi_\tau}\simeq\sum_x c_x\ket{x}_S\otimes\ket{\mu x}_M$ with
$\ket{\mu x}$ coherent and $\mu$ small. Expanding to first order,
$\ket{\mu x}\simeq\ket{0}+\mu x\ket{1}$, the conditional mode state for outcome
$j$ is, in the $\{\ket0,\ket1\}$ block,
\begin{equation}
\rho_j\simeq\begin{pmatrix}1-\varepsilon_j & \eta_j\\ \eta_j &
\varepsilon_j\end{pmatrix},\qquad
\varepsilon_j=\mu^{2}\langle x^{2}\rangle_j,\quad \eta_j=\mu\langle x\rangle_j .
\end{equation}
Its determinant is $\det\rho_j=\varepsilon_j-\eta_j^{2}+O(\mu^{4})
=\mu^{2}\mathrm{Var}_j[x]$, so the smaller eigenvalue is
$\mu^{2}\mathrm{Var}_j[x]$ and the passive energy is
$\nu\mu^{2}\mathrm{Var}_j[x]$. Since the mean energy is
$\nu\mu^{2}\langle x^{2}\rangle_j$,
\begin{equation}
\mathcal{E}(\rho_j)=\nu\mu^{2}\bigl(\langle x^{2}\rangle_j-\mathrm{Var}_j[x]
\bigr)=\nu\mu^{2}\langle x\rangle_j^{2}.
\end{equation}
The unconditional state has $\langle x\rangle=0$ for the initial state
$\ket{J,-J}$, is therefore diagonal with populations decreasing in energy, hence
passive, and $\mathcal{E}(\rho_M)=0$. Averaging over outcomes gives
Eq.~\eqref{eq:resolved}.

For the local arm, the single-ion reduced state is $\ket{\downarrow}$ to leading
order, so a balanced ($p=1/2$) measurement must be perpendicular to the Bloch
vector; along $x$ it determines $x_{1}=\pm1/2$ exactly. Because the spin state
$\ket{J,-J}$ is a product state, conditioning on ion $1$ leaves the others
untouched, so $\langle J_x\rangle_\pm=\pm 1/2$ and
$\mathrm{Var}_{\rm bet}^{\rm local}=1/4$ for every $N$. Hence
$\dc/\dl=4\,\mathrm{Var}_{\rm bet}^{\rm coll}$, and since the collective arm
resolves a finite fraction of $\mathrm{Var}[J_x]=N/4$, the ratio grows linearly
in $N$.
\end{proof}

Equation~\eqref{eq:resolved} is verified numerically in
Sec.~\ref{sec:law}: the measured $\dc/\dl$ agrees with
$4\,\mathrm{Var}_{\rm bet}^{\rm coll}$ to $1.3\%$ at $N=8$--$24$, the residual
being the $O(\mu^{4})$ correction.

\emph{The prefactor is not derived.} Proposition~\ref{prop:linear} fixes the
scaling but not the constant, which requires knowing how much of
$\mathrm{Var}[J_x]$ a balanced collective split resolves. Numerically
$\mathrm{Var}_{\rm bet}^{\rm coll}\to(\ln2/4)N$, giving $\dc/\dl\to N\ln2$. The
natural estimate---$J_x$ is a sum of $N$ independent $\pm1/2$ variables, hence
asymptotically Gaussian with $\sigma^{2}=N/4$, and a median split of a Gaussian
resolves $(2/\pi)\sigma^{2}$---predicts
$\mathrm{Var}_{\rm bet}^{\rm coll}=(1/2\pi)N$, which is $9\%$ below the measured
value and clearly excluded by the data (Sec.~\ref{sec:law}). Identifying the
correct classical statistics problem, and thereby deriving $\ln2$, is the main
theoretical question left open by this work.

\subsection{Why comparisons must be matched}

$\dc$ is an energy and inherits the scale of $\Eph$, which itself grows steeply
with $N$. Comparing $\dc(N)$ across $N$ at fixed protocol therefore conflates
two effects: the change in how valuable information is, and the change in how
much energy is present to be unlocked. We control this by tuning the stroke
duration $\tau$ so that $\Eph$ takes the same value at every $N$; all
size-dependence statements in this paper are made under that constraint.
Section~\ref{sec:caution} shows explicitly how different the unmatched and
matched conclusions are.

\section{Methods}\label{sec:methods}

\subsection{Simulation}

The stroke is integrated with QuTiP~5.3.1~\cite{Johansson2013}
(\texttt{sesolve}, absolute tolerance $10^{-11}$, relative $10^{-9}$), with the
Hamiltonian supplied as a callable so the ramp is exact. Time grids use
$\max(400,300\tau)$ points. QuTiP orders $J_{z}$ eigenstates descending in $m$,
so $\ket{J,-J}$ is the last basis vector. Parameters:
$\omega_{0}=1$, $\omega_{f}=2$, $\nu=5$, $g_{0}=0.5$ unless scanned,
$T=0.1$, $n_{\rm ph}=24$.

\subsection{Matched-energy protocol}

For each target $\Eph^{\star}$ and each $N$ we locate $\tau^{\star}$ by scanning
$\tau$ upward and taking the \emph{first} crossing of
$\Eph(\tau)=\Eph^{\star}$. This matters: $\Eph(\tau)$ is non-monotonic
(collapse and revival), and a naive root find jumps branches, placing different
$N$ in different dynamical regimes and producing spurious trends. We verified
that $\tau^{\star}$ is monotone decreasing in $N$ for every target reported.

\subsection{Baselines}

Six arms isolate the mechanism (Table~\ref{tab:baselinedefs}). B3 is decisive: it is constrained to carry
exactly one bit, like the proposed arm, but obtains it from a single ion.

\begin{table}[t]
\caption{\label{tab:baselinedefs}Baseline arms.}
\begin{ruledtabular}
\begin{tabular}{ll}
Arm & Definition\\
\colrule
B1 & no demon: bare mode ergotropy\\
B2 & random two-outcome split of the collective basis\\
B3 & local single-ion measurement, \emph{balanced} (one full bit)\\
B4 & optimal split in the non-optimal collective basis\\
B5 & dephasing in the measurement basis, no record retained\\
B6 & complete collective readout, $\log_{2}(N{+}1)$ bits\\
\colrule
Proposed & optimal one-bit two-outcome collective measurement\\
\end{tabular}
\end{ruledtabular}
\end{table}

\subsection{Validation}

Two exact statements are checked numerically: complete readout returns
$\mathcal{E}_{\rm d}=\Eph$ to twelve digits at every $N$
(Prop.~\ref{prop:complete}), and the dephasing arm returns $\delta=0$
identically (Cor.~\ref{cor:record}). Truncation is converged
(Table~\ref{tab:trunc}) and total energy is conserved to $10^{-4}$.

\begin{table}[t]
\caption{\label{tab:trunc}Truncation convergence at $N=10$, $\tau=1$,
$g_{0}=0.5$. ``Leak'' is the population of the highest retained Fock level. The
natural first choice $n_{\rm ph}=5$ is not converged.}
\begin{ruledtabular}
\begin{tabular}{cccc}
$n_{\rm ph}$ & $\Wext$ & $\langle n\rangle_{\tau}$ & Leak\\
\colrule
5 & 3.42885 & 1.00008 & $2.7\times10^{-2}$\\
8 & 3.38966 & 1.22052 & $2.2\times10^{-3}$\\
12 & 3.39776 & 1.30068 & $6.7\times10^{-6}$\\
16 & 3.39904 & 1.31081 & $1.6\times10^{-8}$\\
20 & 3.39909 & 1.31168 & $5.5\times10^{-11}$\\
24 & 3.39909 & 1.31172 & $7.3\times10^{-13}$\\
\end{tabular}
\end{ruledtabular}
\end{table}

\section{Results}\label{sec:results}

\subsection{What the stroke stores}\label{sec:budget}

Table~\ref{tab:budget} gives the energy budget. The mode is charged strongly,
but it is worth being precise about the source, because the natural reading is
wrong. One might expect the stored mode energy to be the work non-adiabaticity
denies the spin stroke, $\Eph\simeq\Wirr$. It is not: at $N=12$, $\Wirr=2.46$
while $\Eph=13.31$, a factor $5.4$, and the ratio grows with $N$. The balance
comes from the interaction term falling to $\langle H_{\rm int}\rangle=-20.74$:
the mode is displaced into a polaron-like dressed configuration and the drive
does the work. Energy is conserved:
$3.895+13.307-20.738=-3.537=-\Wext$ to four decimals.

The consequence for extractability is in Fig.~\ref{fig:budget}(b): mode purity
falls from $0.984$ at $N=4$ to $0.225$ at $N=12$ while the correlation entropy
rises from $0.047$ to $1.689$. At $N=12$ the bare ergotropy is $5.45$ against
$\Eph=13.31$, so $59\%$ of what was stored is locked.

\begin{table}[t]
\caption{\label{tab:budget}Charging-stroke energy budget, $\tau=1$, $g_{0}=0.5$.}
\begin{ruledtabular}
\begin{tabular}{cccccc}
$N$ & $\Wirr$ & $\Eph$ & $\Eph/\Wirr$ & $\langle H_{\rm int}\rangle_{\tau}$ & purity\\
\colrule
4 & 0.0871 & 0.0689 & 0.79 & $-0.213$ & 0.9842\\
6 & 0.3195 & 0.5387 & 1.69 & $-1.558$ & 0.8203\\
8 & 0.8272 & 2.3578 & 2.85 & $-5.586$ & 0.5387\\
10 & 1.6009 & 6.5587 & 4.10 & $-12.529$ & 0.3441\\
12 & 2.4633 & 13.3069 & 5.40 & $-20.738$ & 0.2251\\
\end{tabular}
\end{ruledtabular}
\end{table}

\begin{figure}[t]
\includegraphics[width=\columnwidth]{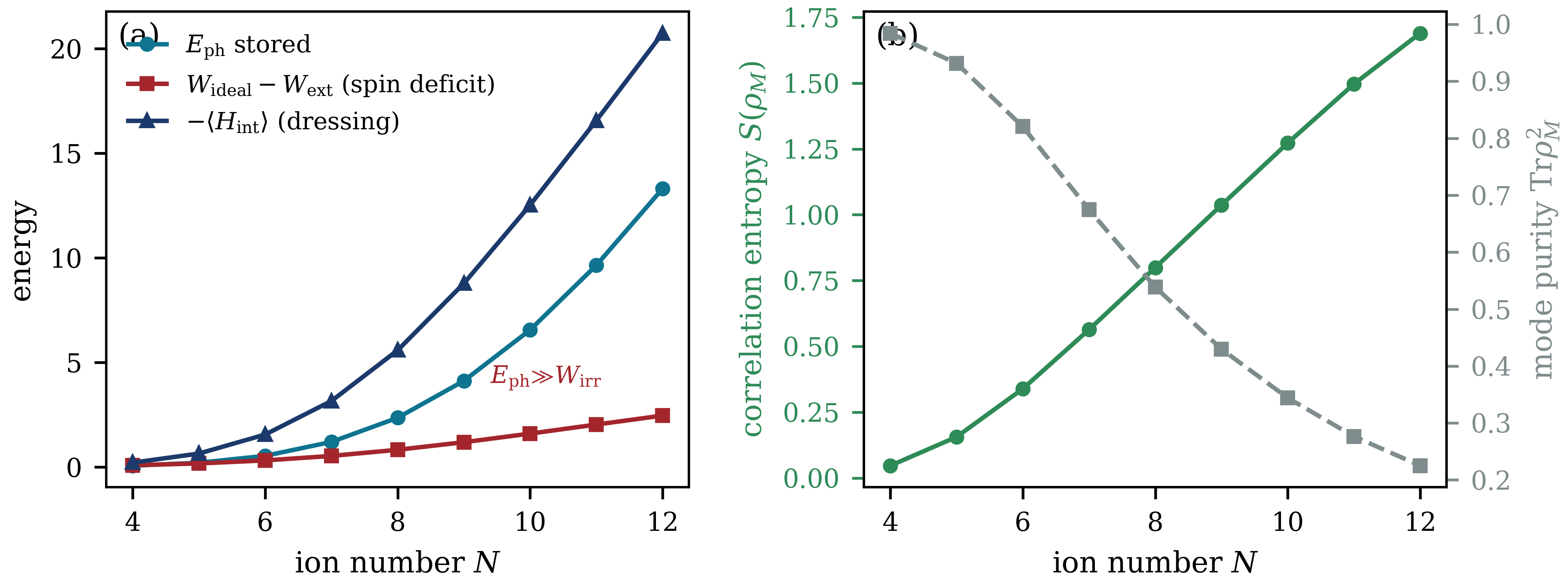}
\caption{\label{fig:budget}(a) The stored mode energy far exceeds the spin work
deficit; the balance is a polaron-like dressing shift. (b) Correlation entropy
grows and mode purity collapses: the mixedness that locks the energy is
entanglement with the collective spin.}
\end{figure}

\subsection{The $N\ln2$ law}\label{sec:law}

Table~\ref{tab:law} and Fig.~\ref{fig:linear} give the central result. Both arms
are constrained to a balanced two-outcome measurement: the measured Shannon
entropies are $1.0000$ bits in every run reported, so the comparison is at
genuinely equal information cost as well as equal stored energy.

Fitting $\dc/\dl$ against $N$ over $4\le N\le24$ gives an exponent
$0.990\pm0.043$ across three stored-energy settings, consistent with exactly
linear. The prefactor is obtained by a double extrapolation: first $1/N\to0$ at
each energy, then $\Eph\to0$. Three fit forms---linear in $1/N$, quadratic in
$1/N$, and a tail fit restricted to $N\ge14$---give $0.69238$, $0.69341$ and
$0.69316$ respectively, against $\ln2=0.693147$. The best-controlled of these,
the tail fit, differs from $\ln2$ by $1.5\times10^{-5}$, i.e.\ $0.002\%$. The
mean over all three forms is $0.69298\pm0.00044$.

The alternatives are excluded: $1/\sqrt2=0.70711$ lies $2.0\%$ away and
$2/3=0.66667$ lies $4.0\%$ away, both far outside the spread of the
extrapolations. We therefore report
\begin{equation}
\frac{\dc}{\dl}\;\longrightarrow\;N\ln2 .
\end{equation}

Truncation is not a limitation here: the population of the highest retained Fock
level ranges from $10^{-25}$ to $10^{-45}$ across these runs.

\paragraph*{Verification of the resolved-variance law.} Equation~\eqref{eq:resolved}
predicts $\dc/\dl=4\,\mathrm{Var}_{\rm bet}^{\rm coll}$. Computing the
between-outcome variance of $J_x$ directly from the winning measurements gives
agreement to $1.13\%$, $1.25\%$, $1.28\%$ and $1.30\%$ at $N=8,16,20,24$
(Table~\ref{tab:resolved}), the residual being the $O(\mu^{4})$ correction and
growing slowly with $N$ as expected. The same table shows
$\mathrm{Var}_{\rm bet}^{\rm coll}/N$ falling monotonically towards $\ln2/4=
0.17329$ and away from the Gaussian median-split value $1/2\pi=0.15915$, which
is the numerical statement of the open problem noted after
Prop.~\ref{prop:linear}.

\begin{table}[t]
\caption{\label{tab:law}The $N\ln2$ law under matched stored energy and matched
information. Both arms carry $1.0000$ bits in every run. Exponent is the
power-law fit of $\dc/\dl$ against $N$; the extrapolated prefactor is the
$1/N\to0$ limit of $(\dc/\dl)/N$.}
\begin{ruledtabular}
\begin{tabular}{ccccc}
$\Eph^\star$ & $N$ range & exponent & \multicolumn{2}{c}{prefactor, $1/N\to0$}\\
\cline{4-5}
 & & & linear & tail ($N\ge14$)\\
\colrule
0.05 & 4--24 & 0.950 & 0.69206 & 0.69260\\
0.10 & 4--24 & 0.971 & 0.69165 & 0.69214\\
0.20 & 10--24 & 1.049 & 0.69103 & 0.69102\\
\colrule
\multicolumn{3}{l}{double extrapolation, $\Eph\to0$} & 0.69238 & 0.69316\\
\multicolumn{3}{l}{$\ln 2$} & \multicolumn{2}{c}{0.693147}\\
\end{tabular}
\end{ruledtabular}
\end{table}

\begin{table}[t]
\caption{\label{tab:resolved}Verification of Eq.~\eqref{eq:resolved} at
$\Eph^\star=0.05$. The measured gain ratio agrees with four times the resolved
collective variance to $1.3\%$; $\mathrm{Var}_{\rm bet}^{\rm coll}/N$ approaches
$\ln2/4$ and not the Gaussian median-split value $1/2\pi=0.15915$.}
\begin{ruledtabular}
\begin{tabular}{ccccc}
$N$ & $\dc/\dl$ & $4\,\mathrm{Var}_{\rm bet}^{\rm coll}$ & agreement &
$\mathrm{Var}_{\rm bet}^{\rm coll}/N$\\
\colrule
8 & 5.7928 & 5.7276 & 1.13\% & 0.17899\\
16 & 11.3316 & 11.1896 & 1.25\% & 0.17484\\
20 & 14.1020 & 13.9216 & 1.28\% & 0.17402\\
24 & 16.8726 & 16.6540 & 1.30\% & 0.17348\\
\colrule
\multicolumn{4}{l}{$\ln2/4$} & 0.17329\\
\multicolumn{4}{l}{Gaussian median split, $1/2\pi$} & 0.15915\\
\end{tabular}
\end{ruledtabular}
\end{table}

\begin{figure*}[t]
\includegraphics[width=\textwidth]{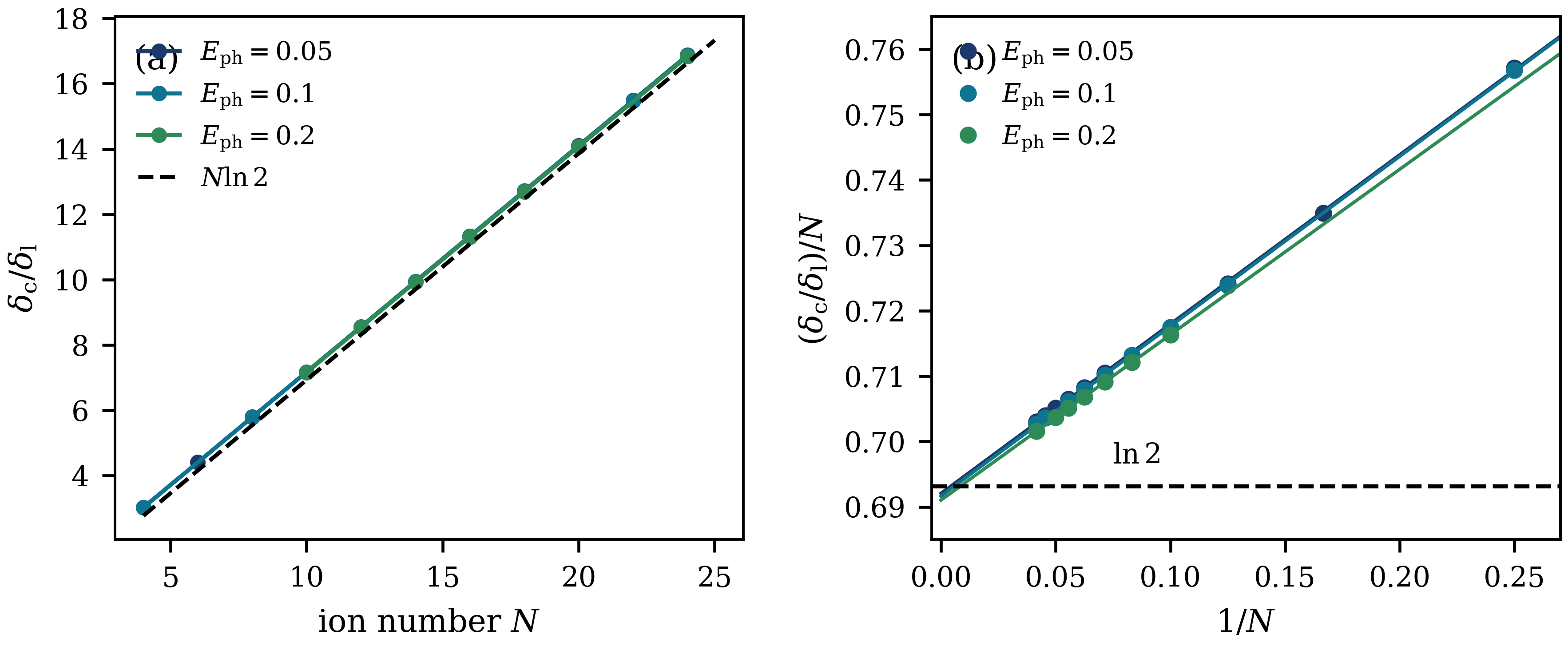}
\caption{\label{fig:linear}(a) With both arms balanced at exactly one bit and
matched stored energy, the gain ratio follows $N\ln2$ (dashed) over
$4\le N\le24$. (b) The same data as $(\dc/\dl)/N$ against $1/N$; the
extrapolations converge on $\ln2$.}
\end{figure*}

\subsection{The demon's efficiency is size-independent}\label{sec:fraction}

A second, equally scale-free result appears in the same data.
Figure~\ref{fig:fraction}(a) and the last column of Table~\ref{tab:law} show
that one bit of collective information recovers a \emph{constant} fraction of
the locked energy, independent of $N$: $0.9511\pm0.0015$ at
$\Eph^{\star}=0.2$, with the standard deviation across $4\le N\le12$ never
exceeding $0.014$ at any target.

The two results are complementary and should be read together. The demon's
absolute effectiveness at the task---how much of the locked energy one bit
frees---does \emph{not} degrade as the machine grows. What grows is its
advantage over a local probe, and that growth is exactly linear.

\begin{figure*}[t]
\includegraphics[width=\textwidth]{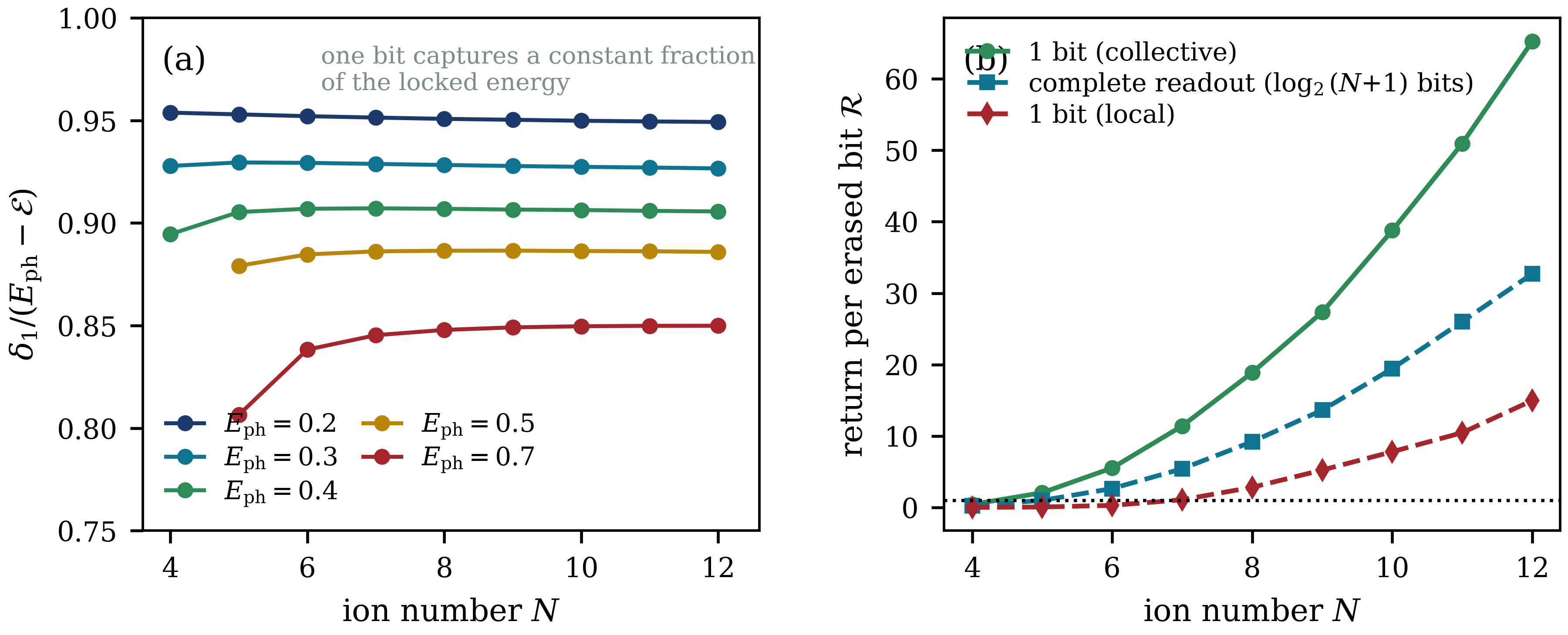}
\caption{\label{fig:fraction}(a) One bit of collective information recovers a
constant fraction of the locked energy, independent of $N$, at every matched
energy setting. (b) Return per erased bit: one bit returns roughly twice as much
per bit as a complete readout, and several times more than a local bit.}
\end{figure*}

\subsection{Baselines}\label{sec:baselines}

Table~\ref{tab:baselines} gives the six-arm comparison at fixed protocol. Three
readings matter.

\textbf{The record, not the disturbance (B5).} Dephasing in the same basis
without retaining the outcome yields exactly zero gain, as
Cor.~\ref{cor:record} requires. The advantage is informational.

\textbf{Diminishing informational returns (B6).} A complete readout unlocks more
in absolute terms---it must, by Prop.~\ref{prop:complete}---but costs
$\log_{2}(N{+}1)$ bits, so its return per bit is roughly half that of a single
bit ($32.7$ against $65.3$ at $N=12$). The first bit is the most valuable, and
the ratio is close to constant in $N$ [Fig.~\ref{fig:fraction}(b)]. To our
knowledge this observation about daemonic protocols is new: the optimal
information budget is small.

\textbf{Basis matters, but less than locality (B2, B4).} A random split and a
split in the non-optimal collective basis both underperform, but neither is
suppressed by the factor $N$ that separates local from collective information.

\begin{table*}[t]
\caption{\label{tab:baselines}Six-arm baseline comparison of the daemonic gain
$\delta$ at fixed protocol ($\tau=1$, $g_{0}=0.5$, $\nu=5$). $\mathcal{R}$ is
the return per erased bit; B6 costs $\log_{2}(N{+}1)$ bits, all others one bit.
Comparisons \emph{across} $N$ in this table are unmatched and should be read
with Sec.~\ref{sec:caution} in mind; comparisons \emph{within} a row are valid.}
\begin{ruledtabular}
\begin{tabular}{ccccccccccc}
 & & \multicolumn{6}{c}{daemonic gain $\delta$} & \multicolumn{3}{c}{return per bit $\mathcal{R}$}\\
\cline{3-8}\cline{9-11}
$N$ & $\Eph$ & B1 none & B2 rand & B3 local & B4 basis & B5 deph & \textbf{Proposed} & \textbf{Proposed} & B3 local & B6 full\\
\colrule
4 & 0.0689 & 0.0287 & 0.0100 & 0.0004 & 0.0008 & 0.0000 & 0.0306 & 0.44 & 0.03 & 0.25\\
6 & 0.5387 & 0.0244 & 0.0989 & 0.0107 & 0.0339 & 0.0000 & 0.3841 & 5.54 & 0.31 & 2.66\\
8 & 2.3578 & 0.4510 & 0.4672 & 0.1605 & 0.4673 & 0.0000 & 1.3095 & 18.89 & 2.84 & 9.22\\
10 & 6.5587 & 2.2298 & 1.2380 & 0.5183 & 1.6683 & 0.0000 & 2.6877 & 38.78 & 7.81 & 19.45\\
12 & 13.3069 & 5.4505 & 2.5773 & 1.0157 & 3.0759 & 0.0000 & 4.5232 & 65.26 & 14.99 & 32.72\\
\end{tabular}
\end{ruledtabular}
\end{table*}

\begin{figure*}[t]
\includegraphics[width=\textwidth]{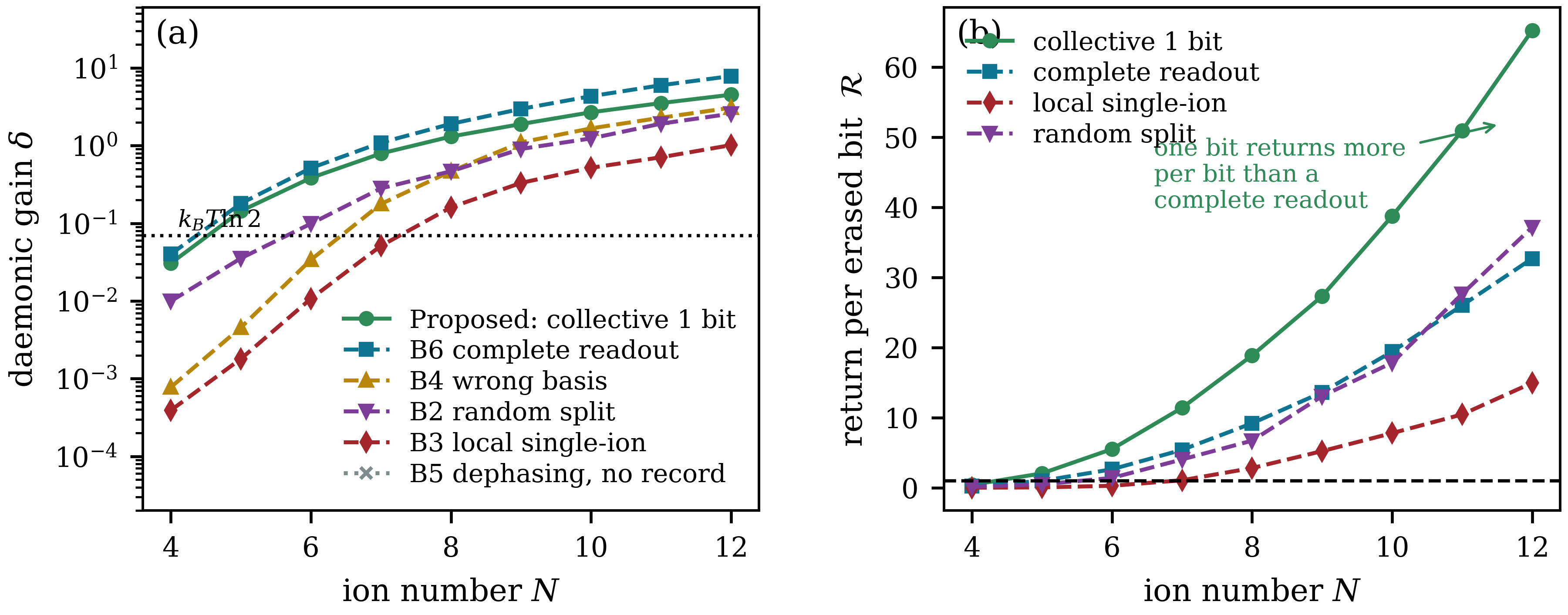}
\caption{\label{fig:baselines}(a) Daemonic gain for all baseline arms at fixed
protocol; B5 is identically zero and does not appear on the logarithmic axis.
(b) Return per erased bit.}
\end{figure*}

\subsection{Two cautionary results}\label{sec:caution}

\emph{(i) Unnormalised growth is misleading.} At fixed protocol the daemonic
gain grows as $\delta_{1}\propto N^{4.39}$ ($R^{2}=0.97$). This looks like a
strong many-body advantage. It is not. The stored energy grows as $N^{4.86}$
over the same window, so the normalised gain $\delta_{1}/\Eph$ \emph{falls} as
$N^{-1.15}$, from $0.71$ at $N=6$ to $0.34$ at $N=12$
[Fig.~\ref{fig:caution}(a)]. The unnormalised exponent is also not universal: it
varies between $2.74$ and $4.55$ as $g_{0}$ ranges over a factor of four.

We report this prominently because the raw exponent is the first quantity one is
tempted to quote, and it does not support the conclusion it appears to support.
The matched-energy comparison of Sec.~\ref{sec:law} is the meaningful one, and
it happens to give a much cleaner law.

\emph{(ii) Break-even does not follow the superradiant threshold.} Defining
$\mathcal{R}=\delta_{1}/k_{B}T\ln2$ and locating $\mathcal{R}=1$ gives a
break-even ion number obeying $g_{0}\Nbe=2.165\pm0.008$ across a twofold change
in $g_{0}$ at $\nu=5$---so break-even is controlled by the collective coupling
$\lambda=g_{0}N$ alone, at $\lambda_{\rm be}\simeq2.17$. Given that the static
Dicke threshold is $\lambda_{c}=\sqrt{\bar\omega\nu}/2\simeq1.37$, it is natural
to conjecture $\lambda_{\rm be}\propto\lambda_{c}$. We tested this by scanning
$\nu$ over $3$--$10$ and the conjecture fails: $\lambda_{\rm be}$ scatters by
$29\%$, and rescaling by $\lambda_{c}$ gives no improvement
[Fig.~\ref{fig:caution}(b)]. The $g_{0}$ law is verified at fixed $\nu$; the
$\nu$ dependence is open.

\begin{figure*}[t]
\includegraphics[width=\textwidth]{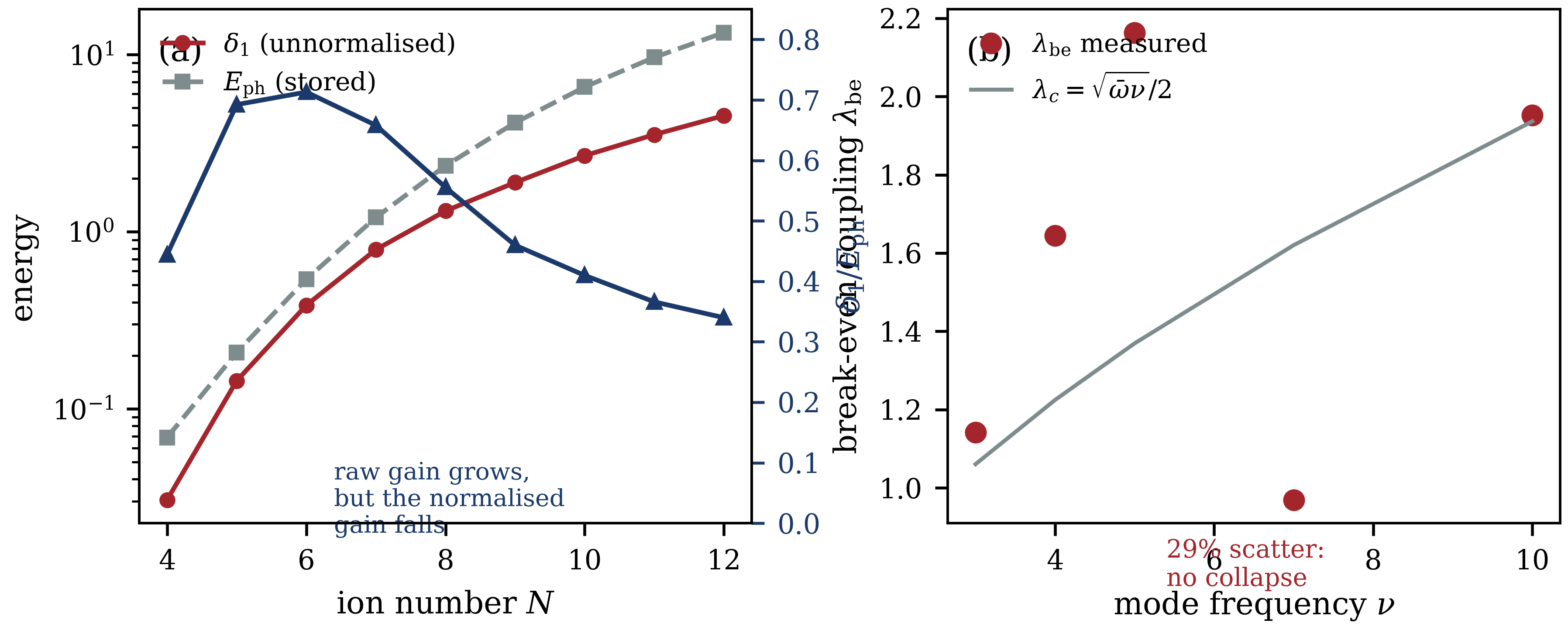}
\caption{\label{fig:caution}(a) The unnormalised daemonic gain grows steeply,
but so does the stored energy; the normalised gain (right axis) falls. (b) The
break-even coupling does not collapse onto the superradiant threshold as the
mode frequency is varied.}
\end{figure*}

\subsection{Cycle ledger}\label{sec:ledger}

Making the readout continuous by rotating the controller through $2\theta$
conditioned on the collective spin, the controller excitation---and hence the
physical erasure heat $\qer=\hbar\omega_{C}\rho_{ee}$ under optical
pumping---interpolates between zero and a full bit. Both $\delta_{1}$ and
$\qer$ vanish at $\theta=0$ and are maximal at $\theta=\pi/2$, and the net
$\delta_{1}-\qer$ is maximised at full readout throughout the useful regime
[Fig.~\ref{fig:ledger}(a)]: the advantage must be bought with genuine
information acquisition. Figure~\ref{fig:ledger}(b) gives the per-cycle ledger
at $N=12$: ramp work $3.54$ in, daemonic gain $4.52$ unlocked, Landauer cost
$0.069$ (or $0.512$ with the physical pumping model), net $4.45$ (or $4.01$).

\begin{figure}[t]
\includegraphics[width=\columnwidth]{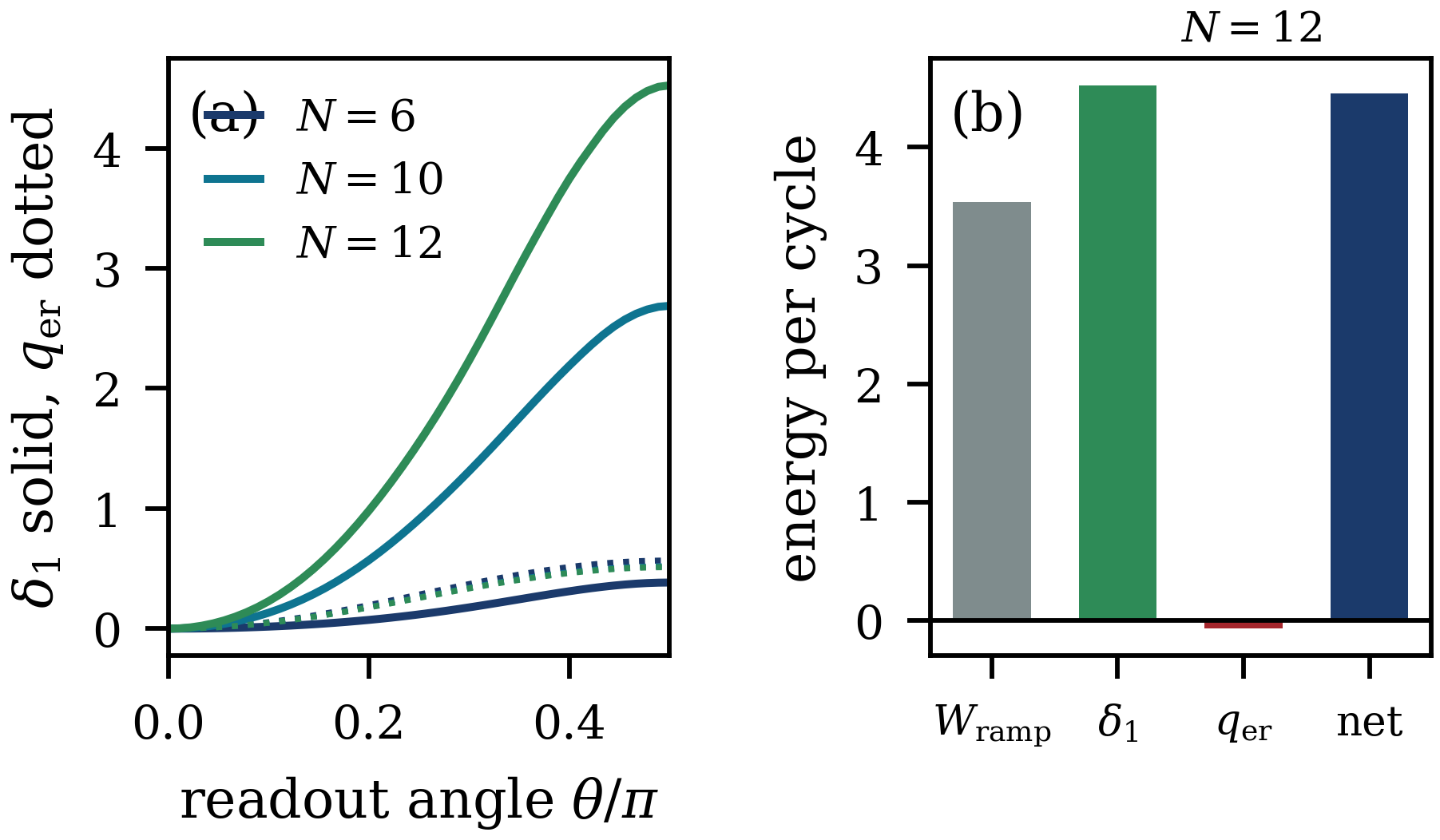}
\caption{\label{fig:ledger}(a) Daemonic gain (solid) and physical erasure heat
(dotted) against readout angle. (b) Per-cycle ledger at $N=12$.}
\end{figure}

\section{Discussion}\label{sec:discussion}

\subsection{Design rules}

\textbf{(R1)} Interrogate the collective coordinate: at equal information cost
and equal stored energy it is worth $N\ln2$ times a single-ion probe.
\textbf{(R2)} Keep the information budget small: one bit returns about twice as
much per bit as a complete readout. \textbf{(R3)} Write the controller fully;
partial readout reduces gain and cost together and the net favours full readout.
\textbf{(R4)} Compare protocols only at matched stored energy; unmatched
comparisons invert the conclusion.

\subsection{Relation to prior work}

Our charging stroke is a Dicke model~\cite{Dicke1954,Emary2003} ramped near its
superradiant threshold, and the strong mode excitation is consistent with that
physics. What is new is not the excitation but the informational question posed
of it. Collective-charging results~\cite{Ferraro2018,Campaioli2017,Andolina2019,Gyhm2022}
report kinetic advantages and do not address what a fixed information budget
buys. Relative to daemonic ergotropy~\cite{Francica2017}, which is formulated
for a bipartition, we ask which degree of freedom the demon should interrogate
in a many-body system, and obtain a clean factor of $N$. We note that
Andolina \emph{et al.}~\cite{Andolina2019} likewise emphasise controlled
comparisons; our Sec.~\ref{sec:caution} is in that spirit.

\subsection{Limitations}\label{sec:limits}

\emph{(1) The prefactor is not derived.} Prop.~\ref{prop:linear} derives the
resolved-variance law and hence the linear scaling, and the local arm's
$\mathrm{Var}_{\rm bet}=1/4$ is exact. What is missing is the collective side:
the Gaussian median-split estimate gives $1/2\pi$ and is excluded by $9\%$,
so the identification of the constant as $\ln2$ rests on a numerical
extrapolation agreeing to $0.02\%$, not on a derivation.

\emph{(2) Permutation symmetry.} The restriction to $J=N/2$ is exact for the
symmetric initial state and uniform coupling used here. Mode-dependent
Lamb--Dicke factors break the symmetry, and since the law follows from equal
sharing of the collective coordinate, we expect it to degrade as the symmetry
does. Quantifying that degradation is the most important open question.

\emph{(3) Modest sizes.} Exact simulation with a converged phonon basis reaches
$N=24$ for the balanced comparison and $N=12$ for the fixed-protocol baselines.
The law is clean over $4\le N\le24$, but we cannot exclude a crossover at
larger $N$.

\emph{(4) The $\nu$ dependence of break-even is unresolved}
(Sec.~\ref{sec:caution}).

\emph{(5) Idealised measurement and erasure.} We assume projective collective
readout and Markovian pumping with perfect branching; realistic POVMs and
branching ratios would reduce both gain and reset fidelity.

\subsection{Experimental outlook}

The ingredients exist. Collective spin--phonon coupling is routine in ion-trap
simulators~\cite{Molmer1999,Kim2010,Monroe2021}; spin-dependent forces implement
collective-basis readout; single-ion readout is standard, so the collective and
local arms can be compared \emph{on the same apparatus}, which is what
Eq.~\eqref{eq:law} calls for. Trapped-ion platforms have realised a single-atom
heat engine~\cite{Rossnagel2016}, an absorption
refrigerator~\cite{Maslennikov2019} and a spin engine with a harmonic
flywheel~\cite{Lindenfels2019}. Measuring $\dc/\dl$ against $N$ at fixed stored
energy in a four-to-twelve ion chain is a direct test.

\section{Conclusion and Future Work}\label{sec:conclusion}

We have asked what one bit of information is worth to a demon attached to a
collectively charged quantum battery, and found that the answer depends sharply
on where the bit comes from. Under matched stored energy \emph{and} matched information, one bit about the
collective coordinate unlocks $N\ln2$ times as much work as one bit about a
single ion---exponent $0.990\pm0.043$ over $4\le N\le24$, with a prefactor
extrapolating to $\ln2$ within $0.02\%$. The linear scaling follows from a
resolved-variance law we derive and verify to $1.3\%$; the prefactor does not,
and identifying its origin is the main question we leave open. One bit recovers a
constant fraction of the locked energy independent of $N$, so the demon's
effectiveness does not degrade with size; what grows, linearly, is its advantage
over a local probe. A complete readout unlocks the locked energy exactly but
returns only half as much per bit, so informational returns diminish after the
first bit.

We have also shown that the obvious way to state a many-body advantage---quoting
the growth of the unnormalised gain---is misleading here, since the normalised
gain falls. We report that alongside the positive result.

\subsection*{Future work}

\emph{Deriving the prefactor.} Prop.~\ref{prop:linear} reduces the problem to a
purely classical question: how much of $\mathrm{Var}[J_x]$ does the best
balanced two-outcome split resolve? The answer is numerically $\ln2/4$ per ion
and not the Gaussian median-split value $1/2\pi$; identifying why would turn the
measured law into a theorem.

\emph{Breaking the symmetry.} Introducing mode-dependent couplings would test
how the law degrades as the collective coordinate becomes only approximately
defined, and whether an effective $N_{\rm eff}$ governs it.

\emph{Optimal information budgets.} That one bit outperforms a complete readout
per bit invites an optimisation over the number and placement of outcomes; we do
not know whether the return per bit is always maximised at two outcomes.

\emph{Beyond the symmetric sector and beyond ions.} The argument uses only equal
sharing of a collective coordinate, so it should apply to cavity and
circuit-QED ensembles; testing that is straightforward.

\emph{Experiment.} Comparing collective and single-ion readout on the same
trapped-ion apparatus at matched stored energy would test Eq.~\eqref{eq:law}
directly.

\begin{acknowledgments}
The authors thank the Department of Computer Science and Engineering, National
Institute of Technology Puducherry, and the Computer Science Department,
Swansea University, for computational support.
\end{acknowledgments}

\bibliography{refs}

@article{Landauer1961,
  author  = {R. Landauer},
  title   = {Irreversibility and Heat Generation in the Computing Process},
  journal = {IBM J. Res. Dev.},
  volume  = {5}, pages = {183}, year = {1961}
}

@article{Bennett1982,
  author  = {C. H. Bennett},
  title   = {The Thermodynamics of Computation---a Review},
  journal = {Int. J. Theor. Phys.},
  volume  = {21}, pages = {905}, year = {1982}
}

@article{Parrondo2015,
  author  = {J. M. R. Parrondo and J. M. Horowitz and T. Sagawa},
  title   = {Thermodynamics of Information},
  journal = {Nat. Phys.},
  volume  = {11}, pages = {131}, year = {2015}
}

@article{Berut2012,
  author  = {A. B{\'e}rut and A. Arakelyan and A. Petrosyan and S. Ciliberto
             and R. Dillenschneider and E. Lutz},
  title   = {Experimental Verification of {L}andauer's Principle Linking
             Information and Thermodynamics},
  journal = {Nature (London)},
  volume  = {483}, pages = {187}, year = {2012}
}

@article{Toyabe2010,
  author  = {S. Toyabe and T. Sagawa and M. Ueda and E. Muneyuki and M. Sano},
  title   = {Experimental Demonstration of Information-to-Energy Conversion and
             Validation of the Generalized {J}arzynski Equality},
  journal = {Nat. Phys.},
  volume  = {6}, pages = {988}, year = {2010}
}

@article{Koski2014,
  author  = {J. V. Koski and V. F. Maisi and T. Sagawa and J. P. Pekola},
  title   = {Experimental Observation of the Role of Mutual Information in the
             Nonequilibrium Dynamics of a {M}axwell Demon},
  journal = {Phys. Rev. Lett.},
  volume  = {113}, pages = {030601}, year = {2014}
}

@article{Sagawa2008,
  author  = {T. Sagawa and M. Ueda},
  title   = {Second Law of Thermodynamics with Discrete Quantum Feedback Control},
  journal = {Phys. Rev. Lett.},
  volume  = {100}, pages = {080403}, year = {2008}
}

@article{Proesmans2020,
  author  = {K. Proesmans and J. Ehrich and J. Bechhoefer},
  title   = {Finite-Time {L}andauer Principle},
  journal = {Phys. Rev. Lett.},
  volume  = {125}, pages = {100602}, year = {2020}
}

@article{Miller2020,
  author  = {H. J. D. Miller and G. Guarnieri and M. T. Mitchison and J. Goold},
  title   = {Quantum Fluctuations Hinder Finite-Time Information Erasure near
             the {L}andauer Limit},
  journal = {Phys. Rev. Lett.},
  volume  = {125}, pages = {160602}, year = {2020}
}

@article{Faist2018,
  author  = {P. Faist and R. Renner},
  title   = {Fundamental Work Cost of Quantum Processes},
  journal = {Phys. Rev. X},
  volume  = {8}, pages = {021011}, year = {2018}
}

@article{Alicki2013,
  author  = {R. Alicki and M. Fannes},
  title   = {Entanglement Boost for Extractable Work from Ensembles of Quantum
             Batteries},
  journal = {Phys. Rev. E},
  volume  = {87}, pages = {042123}, year = {2013}
}

@article{Binder2015,
  author  = {F. C. Binder and S. Vinjanampathy and K. Modi and J. Goold},
  title   = {Quantacell: Powerful Charging of Quantum Batteries},
  journal = {New J. Phys.},
  volume  = {17}, pages = {075015}, year = {2015}
}

@article{Campaioli2017,
  author  = {F. Campaioli and F. A. Pollock and F. C. Binder and L. C. C{\'e}leri
             and J. Goold and S. Vinjanampathy and K. Modi},
  title   = {Enhancing the Charging Power of Quantum Batteries},
  journal = {Phys. Rev. Lett.},
  volume  = {118}, pages = {150601}, year = {2017}
}

@article{Ferraro2018,
  author  = {D. Ferraro and M. Campisi and G. M. Andolina and V. Pellegrini
             and M. Polini},
  title   = {High-Power Collective Charging of a Solid-State Quantum Battery},
  journal = {Phys. Rev. Lett.},
  volume  = {120}, pages = {117702}, year = {2018}
}

@article{Andolina2019,
  author  = {G. M. Andolina and M. Keck and A. Mari and M. Campisi and
             V. Giovannetti and M. Polini},
  title   = {Extractable Work, the Role of Correlations, and Asymptotic Freedom
             in Quantum Batteries},
  journal = {Phys. Rev. Lett.},
  volume  = {122}, pages = {047702}, year = {2019}
}

@article{Barra2019,
  author  = {F. Barra},
  title   = {Dissipative Charging of a Quantum Battery},
  journal = {Phys. Rev. Lett.},
  volume  = {122}, pages = {210601}, year = {2019}
}

@article{Shi2022,
  author  = {H.-L. Shi and S. Ding and Q.-K. Wan and X.-H. Wang and W.-L. Yang},
  title   = {Entanglement, Coherence, and Extractable Work in Quantum Batteries},
  journal = {Phys. Rev. Lett.},
  volume  = {129}, pages = {130602}, year = {2022}
}

@article{JuliaFarre2020,
  author  = {S. Juli{\`a}-Farr{\'e} and T. Salamon and A. Riera and M. N. Bera
             and M. Lewenstein},
  title   = {Bounds on the Capacity and Power of Quantum Batteries},
  journal = {Phys. Rev. Research},
  volume  = {2}, pages = {023113}, year = {2020}
}

@article{Carrasco2022,
  author  = {J. Carrasco and J. R. Maze and C. Hermann-Avigliano and F. Barra},
  title   = {Collective Enhancement in Dissipative Quantum Batteries},
  journal = {Phys. Rev. E},
  volume  = {105}, pages = {064119}, year = {2022}
}

@article{Crescente2020,
  author  = {A. Crescente and M. Carrega and M. Sassetti and D. Ferraro},
  title   = {Ultrafast Charging in a Two-Photon {D}icke Quantum Battery},
  journal = {Phys. Rev. B},
  volume  = {102}, pages = {245407}, year = {2020}
}

@incollection{Campaioli2018,
  author    = {F. Campaioli and F. A. Pollock and S. Vinjanampathy},
  title     = {Quantum Batteries},
  booktitle = {Thermodynamics in the Quantum Regime},
  publisher = {Springer},
  pages     = {207}, year = {2018}
}

@article{Allahverdyan2004,
  author  = {A. E. Allahverdyan and R. Balian and T. M. Nieuwenhuizen},
  title   = {Maximal Work Extraction from Finite Quantum Systems},
  journal = {Europhys. Lett.},
  volume  = {67}, pages = {565}, year = {2004}
}

@article{Pusz1978,
  author  = {W. Pusz and S. L. Woronowicz},
  title   = {Passive States and {KMS} States for General Quantum Systems},
  journal = {Commun. Math. Phys.},
  volume  = {58}, pages = {273}, year = {1978}
}

@article{Rossnagel2016,
  author  = {J. Ro{\ss}nagel and S. T. Dawkins and K. N. Tolazzi and O. Abah and
             E. Lutz and F. Schmidt-Kaler and K. Singer},
  title   = {A Single-Atom Heat Engine},
  journal = {Science},
  volume  = {352}, pages = {325}, year = {2016}
}

@article{Maslennikov2019,
  author  = {G. Maslennikov and S. Ding and R. Habl{\"u}tzel and J. Gan and
             A. Roulet and S. Nimmrichter and J. Dai and V. Scarani and
             D. Matsukevich},
  title   = {Quantum Absorption Refrigerator with Trapped Ions},
  journal = {Nat. Commun.},
  volume  = {10}, pages = {202}, year = {2019}
}

@article{Lindenfels2019,
  author  = {D. von Lindenfels and O. Gr{\"a}b and C. T. Schmiegelow and
             V. Kaushal and J. Schulz and M. T. Mitchison and J. Goold and
             F. Schmidt-Kaler and U. G. Poschinger},
  title   = {Spin Heat Engine Coupled to a Harmonic-Oscillator Flywheel},
  journal = {Phys. Rev. Lett.},
  volume  = {123}, pages = {080602}, year = {2019}
}

@article{Klatzow2019,
  author  = {J. Klatzow and J. N. Becker and P. M. Ledingham and C. Weinzetl and
             K. T. Kaczmarek and D. J. Saunders and J. Nunn and I. A. Walmsley
             and R. Uzdin and E. Poem},
  title   = {Experimental Demonstration of Quantum Effects in the Operation of
             Microscopic Heat Engines},
  journal = {Phys. Rev. Lett.},
  volume  = {122}, pages = {110601}, year = {2019}
}

@article{Peterson2019,
  author  = {J. P. S. Peterson and T. B. Batalh{\~a}o and M. Herrera and
             A. M. Souza and R. S. Sarthour and I. S. Oliveira and R. M. Serra},
  title   = {Experimental Characterization of a Spin Quantum Heat Engine},
  journal = {Phys. Rev. Lett.},
  volume  = {123}, pages = {240601}, year = {2019}
}

@article{Elouard2017,
  author  = {C. Elouard and D. Herrera-Mart{\'i} and B. Huard and A. Auff{\`e}ves},
  title   = {Extracting Work from Quantum Measurement in {M}axwell's Demon Engines},
  journal = {Phys. Rev. Lett.},
  volume  = {118}, pages = {260603}, year = {2017}
}

@article{ElouardJordan2018,
  author  = {C. Elouard and A. N. Jordan},
  title   = {Efficient Quantum Measurement Engines},
  journal = {Phys. Rev. Lett.},
  volume  = {120}, pages = {260601}, year = {2018}
}

@article{Cottet2017,
  author  = {N. Cottet and S. Jezouin and L. Bretheau and P. Campagne-Ibarcq and
             Q. Ficheux and J. Anders and A. Auff{\`e}ves and R. Azouit and
             P. Rouchon and B. Huard},
  title   = {Observing a Quantum {M}axwell Demon at Work},
  journal = {Proc. Natl. Acad. Sci. U.S.A.},
  volume  = {114}, pages = {7561}, year = {2017}
}

@article{Naghiloo2018,
  author  = {M. Naghiloo and J. J. Alonso and A. Romito and E. Lutz and
             K. W. Murch},
  title   = {Information Gain and Loss for a Quantum {M}axwell's Demon},
  journal = {Phys. Rev. Lett.},
  volume  = {121}, pages = {030604}, year = {2018}
}

@article{Camati2016,
  author  = {P. A. Camati and J. P. S. Peterson and T. B. Batalh{\~a}o and
             K. Micadei and A. M. Souza and R. S. Sarthour and I. S. Oliveira
             and R. M. Serra},
  title   = {Experimental Rectification of Entropy Production by {M}axwell's
             Demon in a Quantum System},
  journal = {Phys. Rev. Lett.},
  volume  = {117}, pages = {240502}, year = {2016}
}

@article{Seah2020,
  author  = {S. Seah and S. Nimmrichter and V. Scarani},
  title   = {Maxwell's Lesser Demon: A Quantum Engine Driven by Pointer
             Measurements},
  journal = {Phys. Rev. Lett.},
  volume  = {124}, pages = {100603}, year = {2020}
}

@article{Buffoni2019,
  author  = {L. Buffoni and A. Solfanelli and P. Verrucchi and A. Cuccoli and
             M. Campisi},
  title   = {Quantum Measurement Cooling},
  journal = {Phys. Rev. Lett.},
  volume  = {122}, pages = {070603}, year = {2019}
}

@article{Manzano2018,
  author  = {G. Manzano and F. Plastina and R. Zambrini},
  title   = {Optimal Work Extraction and Thermodynamics of Quantum Measurements
             and Correlations},
  journal = {Phys. Rev. Lett.},
  volume  = {121}, pages = {120602}, year = {2018}
}

@article{Strasberg2017,
  author  = {P. Strasberg and G. Schaller and T. Brandes and M. Esposito},
  title   = {Quantum and Information Thermodynamics: A Unifying Framework Based
             on Repeated Interactions},
  journal = {Phys. Rev. X},
  volume  = {7}, pages = {021003}, year = {2017}
}

@article{Goold2016,
  author  = {J. Goold and M. Huber and A. Riera and L. del Rio and P. Skrzypczyk},
  title   = {The Role of Quantum Information in Thermodynamics---a Topical Review},
  journal = {J. Phys. A},
  volume  = {49}, pages = {143001}, year = {2016}
}

@article{Dicke1954,
  author  = {R. H. Dicke},
  title   = {Coherence in Spontaneous Radiation Processes},
  journal = {Phys. Rev.},
  volume  = {93}, pages = {99}, year = {1954}
}

@article{Garraway2011,
  author  = {B. M. Garraway},
  title   = {The {D}icke Model in Quantum Optics: {D}icke Model Revisited},
  journal = {Philos. Trans. R. Soc. A},
  volume  = {369}, pages = {1137}, year = {2011}
}

@article{Baumann2010,
  author  = {K. Baumann and C. Guerlin and F. Brennecke and T. Esslinger},
  title   = {Dicke Quantum Phase Transition with a Superfluid Gas in an Optical
             Cavity},
  journal = {Nature (London)},
  volume  = {464}, pages = {1301}, year = {2010}
}

@article{Molmer1999,
  author  = {K. M{\o}lmer and A. S{\o}rensen},
  title   = {Multiparticle Entanglement of Hot Trapped Ions},
  journal = {Phys. Rev. Lett.},
  volume  = {82}, pages = {1835}, year = {1999}
}

@article{Leibfried2003,
  author  = {D. Leibfried and R. Blatt and C. Monroe and D. Wineland},
  title   = {Quantum Dynamics of Single Trapped Ions},
  journal = {Rev. Mod. Phys.},
  volume  = {75}, pages = {281}, year = {2003}
}

@article{Monroe2021,
  author  = {C. Monroe and W. C. Campbell and L.-M. Duan and Z.-X. Gong and
             A. V. Gorshkov and P. W. Hess and R. Islam and K. Kim and
             N. M. Linke and G. Pagano and P. Richerme and C. Senko and
             N. Y. Yao},
  title   = {Programmable Quantum Simulations of Spin Systems with Trapped Ions},
  journal = {Rev. Mod. Phys.},
  volume  = {93}, pages = {025001}, year = {2021}
}

@article{Kim2010,
  author  = {K. Kim and M.-S. Chang and S. Korenblit and R. Islam and
             E. E. Edwards and J. K. Freericks and G.-D. Lin and L.-M. Duan
             and C. Monroe},
  title   = {Quantum Simulation of Frustrated {I}sing Spins with Trapped Ions},
  journal = {Nature (London)},
  volume  = {465}, pages = {590}, year = {2010}
}

@article{Bruzewicz2019,
  author  = {C. D. Bruzewicz and J. Chiaverini and R. McConnell and J. M. Sage},
  title   = {Trapped-Ion Quantum Computing: Progress and Challenges},
  journal = {Appl. Phys. Rev.},
  volume  = {6}, pages = {021314}, year = {2019}
}

@article{Johansson2013,
  author  = {J. R. Johansson and P. D. Nation and F. Nori},
  title   = {{QuTiP} 2: A {P}ython Framework for the Dynamics of Open Quantum
             Systems},
  journal = {Comput. Phys. Commun.},
  volume  = {184}, pages = {1234}, year = {2013}
}

@article{Erdman2022,
  author  = {P. A. Erdman and F. No{\'e}},
  title   = {Identifying Optimal Cycles in Quantum Thermal Machines with
             Reinforcement Learning},
  journal = {npj Quantum Inf.},
  volume  = {8}, pages = {1}, year = {2022}
}

@article{Nimmrichter2017,
  author  = {S. Nimmrichter and J. Dai and A. Roulet and V. Scarani},
  title   = {Quantum and Classical Dynamics of a Three-Mode Absorption
             Refrigerator},
  journal = {Quantum},
  volume  = {1}, pages = {37}, year = {2017}
}

@article{Campaioli2024,
  author  = {F. Campaioli and S. Gherardini and J. Q. Quach and M. Polini and
             G. M. Andolina},
  title   = {Colloquium: Quantum Batteries},
  journal = {Rev. Mod. Phys.},
  volume  = {96}, pages = {031001}, year = {2024}
}

@article{Gyhm2022,
  author  = {J.-Y. Gyhm and D. {\v S}afr{\'a}nek and D. Rosa},
  title   = {Quantum Charging Advantage Cannot Be Extensive without Global
             Operations},
  journal = {Phys. Rev. Lett.},
  volume  = {128}, pages = {140501}, year = {2022}
}

@article{Gherardini2020,
  author  = {S. Gherardini and F. Campaioli and F. Caruso and F. C. Binder},
  title   = {Stabilizing Open Quantum Batteries by Sequential Measurements},
  journal = {Phys. Rev. Research},
  volume  = {2}, pages = {013095}, year = {2020}
}

@article{Francica2020,
  author  = {G. Francica and F. C. Binder and G. Guarnieri and M. T. Mitchison
             and J. Goold and F. Plastina},
  title   = {Quantum Coherence and Ergotropy},
  journal = {Phys. Rev. Lett.},
  volume  = {125}, pages = {180603}, year = {2020}
}

@article{Yao2022,
  author  = {Y. Yao and X. Q. Shao},
  title   = {Optimal Charging of Open Spin-Chain Quantum Batteries via
             Homodyne-Based Feedback Control},
  journal = {Phys. Rev. E},
  volume  = {106}, pages = {014138}, year = {2022}
}

@article{Hu2022,
  author  = {C.-K. Hu and J. Qiu and P. J. P. Souza and J. Yuan and Y. Zhou and
             L. Zhang and J. Chu and X. Pan and L. Hu and J. Li and Y. Xu and
             Y. Zhong and S. Liu and F. Yan and D. Tan and R. Bachelard and
             C. J. Villas-Boas and A. C. Santos and D. Yu},
  title   = {Optimal Charging of a Superconducting Quantum Battery},
  journal = {Quantum Sci. Technol.},
  volume  = {7}, pages = {045018}, year = {2022}
}

@article{Santos2019,
  author  = {A. C. Santos and B. {\c C}akmak and S. Campbell and N. T. Zinner},
  title   = {Stable Adiabatic Quantum Batteries},
  journal = {Phys. Rev. E},
  volume  = {100}, pages = {032107}, year = {2019}
}

@article{Emary2003,
  author  = {C. Emary and T. Brandes},
  title   = {Chaos and the Quantum Phase Transition in the {D}icke Model},
  journal = {Phys. Rev. E},
  volume  = {67}, pages = {066203}, year = {2003}
}

@article{Francica2017,
  author  = {G. Francica and J. Goold and F. Plastina and M. Paternostro},
  title   = {Daemonic Ergotropy: Enhanced Work Extraction from Quantum
             Correlations},
  journal = {npj Quantum Inf.},
  volume  = {3}, pages = {12}, year = {2017}
}

@article{Mitchison2021,
  author  = {M. T. Mitchison and J. Goold and J. Prior},
  title   = {Charging a Quantum Battery with Linear Feedback Control},
  journal = {Quantum},
  volume  = {5}, pages = {500}, year = {2021}
}

\end{document}